\documentclass{article}
\usepackage{enumerate}
\usepackage{url}
\usepackage{cite}
\usepackage{footnote}
\usepackage{amsmath}
\usepackage{amssymb}
\usepackage{amsthm}
\newtheorem{theorem}{Theorem}
\newtheorem{lemma}[theorem]{Lemma}
\usepackage{graphicx}
\usepackage{xcolor}
\definecolor{light_gray}{RGB}{225,225,225}
\usepackage{subcaption}

\usepackage{float}
\usepackage{bm}
\usepackage{multirow}
\usepackage[plain]{algorithm}
\usepackage[noend]{algorithmic}
\usepackage{tikz}
\usetikzlibrary{fit, decorations.pathreplacing,calc}
\newcommand{\tikzmark}[1]{\tikz[overlay,remember picture] \node (#1) {};}

\newcommand*{\addNote}[5]{%
    \tikz[remember picture,overlay]{
        \node (x) at ($(#3)+(0.3,0)$) {};
        \node (y1) at ($(#1)+(0,0.5\baselineskip)$) {};
        \node (y2) at ($(#2)-(0,0.25\baselineskip)$) {};
        \draw [decoration={brace,amplitude=#4},decorate,ultra thick]
            (x |- y1) -- (x |- y2)
                node [align=center, text width=1cm, pos=0.5, anchor=west] {#5};
    }
}%
\newcommand{\bigbox}[3]{%
    \tikz[remember picture,overlay]{
        \node (x2) at ($(#2)+(0.1,0)$) {};
        \node (y1) at ($(#3)-(0,0.75\baselineskip)$) {};
        \node (y2) at ($(#3)-(0,4.25\baselineskip)$) {};
        \node[yshift=2,fill=light_gray,fit={(#1 |- y1)(x2 |- y2)}] {};
    }
    \ignorespaces
}%
\newcommand{\smallbox}[3]{%
    \tikz[remember picture,overlay]{
        \node (x2) at ($(#2)+(0.1,0)$) {};
        \node (y1) at ($(#3)-(0,0.75\baselineskip)$) {};
        \node (y2) at ($(#3)-(0,1.25\baselineskip)$) {};
        \node[yshift=2,fill=light_gray,fit={(#1 |- y1)(x2 |- y2)}] {};
    }
    \ignorespaces
}%
\begin{document}

\title{A Simple Balanced Search Tree with\\
No Balance Criterion}

\author{Tae Woo Kim}
\date{
School of Computing\\
Korea Advanced Institute of Science and Technology\\
travis1829@kaist.ac.kr
}

\maketitle

\begin{abstract}
We present a method that maintains a balanced binary search tree without using any tree balance criterion at all, with the ultimate aim of maximum simplicity. In fact, our method is highly intuitive, and we only need to add minimal extra code and a simple partial-rebuilding algorithm to a naive binary search tree. Our method will be suitable as a highly simple and short solution when amortized logarithmic costs are enough.
\end{abstract}

\section{Introduction}

The balanced binary search tree (or balanced BST) is one of the most fundamental data structures. Though there are many ways to implement the dictionary, we use the balanced BST when we need logarithmic worst-case costs. Logarithmic worst-case cost is the key feature of balanced BSTs, and it distinguishes a balanced BST from a naive BST, where the former uses a tree balancing method to guarantee logarithmic costs. Common examples of a balanced BST include AVL trees \cite{Adelson-VelskyAVL}, binary B-trees \cite{BayerBB} or its simplified variant \cite{AnderssonAA}, weight-balanced trees \cite{NievergeltWB} or its variant \cite{Overmars:1987:DDD:535500}, and general balanced trees \cite{AnderssonGBT}.

However, it is a pity that such a fundamental data structure is infamous for being cumbersome to implement. To cite Munro, Papadakis, and Sedgewick \cite{MunroDSL}, balanced BSTs use a notorious balancing algorithm that needs to consider numerous cases, and hence, implementations are said to be too complicated for average programmers. Not to mention that balanced BSTs are often neglected, or to cite Andersson \cite{AnderssonAA}, they are often replaced by poor methods instead. Andersson \cite{AnderssonGBT} also said in another paper that many of the commonly used balanced BSTs use a complicated balance criterion to restrict trees and detect imbalance. Also, note that especially nowadays, computer science is needed and used by so many people from all kinds of fields. These strongly emphasize the importance of a simple and explicit implementation.

Fortunately, we have thankful research related to this well-known issue. People searched for a simpler balance criterion. For example, Arne Andersson \cite{AnderssonAA} used bookmarking to reduce the cases that need to be considered in binary B-trees \cite{BayerBB}, and this led to a generalized tree balance criterion with generalized rotation procedures. Also, the general balanced tree \cite{AnderssonGBT} showed that we only need a global balance criterion instead of more complicated ones, and it is mentioned that this maybe-simplest criterion is all we need if amortized logarithmic costs are enough.

However, we go further than trying to get a simpler balance criterion and show that we actually do not need any balance criterion at all if (again) amortized logarithmic costs are enough. Importantly, the ultimate aim is a tree balancing method with maximum simplicity, and we will show that our method has a highly simplified concept and implementation. This is obtained from the following observations that we will explain later in detail.
\begin{enumerate}[$\bullet$]
\item The concept of ``detecting imbalance and then appropriately rebalancing'' can be replaced by ``scheduled rebuilding.''
\item Scheduling rebuilds can be easily done by bookmarking $O(\log n)$ bits on each node and adding simple, minimal code to common implementations of tree procedures.
\item We can maintain a tree using the exact same balancing method for both inserts and deletes.
\item By using tree rebuilds instead of tree rotations, we can offer a more abstract and simple balancing method. Also, rebuilds can be easily implemented by only using a basic tree traversal and a simple recursive algorithm.
\end{enumerate}

In Section~\ref{sec:operations}, we explain our method and see how we can implement our method by only adding minimal code and a simple partial-rebuilding algorithm to a naive BST. Next, in Section~\ref{sec:proof}, we show that our tree has a logarithmic height, and that insert/delete operations have an amortized logarithmic cost. Since we do not use a balance criterion, our tree's shape is flexible, and hence, we use a different analysis compared to other most trees. Finally, in Section~\ref{sec:discussions}, we discuss our tree.

\section {Notations and Notes}
In this paper, we use the following notations.

For a node $node$,
\begin{enumerate}[$\bullet$]
    \item $node.key$ is the key stored at $node$.
    \item $node.left$ is the left child node of $node$.
    \item $node.right$ is the right child node of $node$.
    \item $node.timer$ is the \emph{timer}\footnote{Definition of a \emph{timer} is explained in Section~\ref{sec:operations}.} stored at $node$. 
\end{enumerate}

Also, for a subtree $T$, $T.size$ is the size of $T$.

Note that in the following, when we use the word ``tree,'' we refer to the entire tree, not just a part of it.

\section{The Balancing Method}
\label{sec:operations}
Similar to general balanced trees or the variant of weight-balanced trees \cite{Overmars:1987:DDD:535500}, our tree occasionally rebuilds a subtree. Here, rebuilding a subtree means balancing it into a perfectly balanced one, and this is explained more in Section~\ref{sec:rebuild}. In these trees, we rebuild a subtree when it does not satisfy the balance criterion. However, in our method, we schedule rebuilds beforehand instead.

As a start, we briefly explain our method in one sentence.
\newline\newline
\textit{After rebuilding a subtree of size $n$, we rebuild it again after inserting or deleting a node in it $\lfloor kn \rfloor$ times, where $0<k<1$ is a constant.}\newline
%$\max\left(1,\left\lfloor kn\right\rfloor\right)$

As a simple way to implement our method, we let each subtree have a \emph{timer} integer stored at its root to schedule rebuilds. The \emph{timer} of a subtree stores the remaining number of insertions/deletions until rebuilding it. When we rebuild a subtree, we reset the \emph{timer} for all of its nodes. Specifically, for the root of a subtree of size $n$, we reset the \emph{timer} to $\max\left(1,\left\lfloor kn\right\rfloor\right)$. Note that we use $\max\left(1,\left\lfloor kn\right\rfloor\right)$ instead of just $\lfloor kn \rfloor$ to make the \emph{timer} be at least 1. As its name tells, we decrease the \emph{timer} by 1 each time we insert or delete a node in the subtree, and if the \emph{timer} reaches 0, we rebuild the subtree. Additionally, for newly inserted nodes, we set the \emph{timer} to 1.

After inserting or deleting a node, we decrease the \emph{timer} of some nodes, and then, we may have more than one node whose \emph{timer} reached 0. In that case, we only need to rebuild the subtree rooted by the one with the least depth. We can easily see that such subtree will include all such nodes and their rooted subtrees.

Using these \emph{timer} values, we just need to do the following.
\begin{enumerate}[$\bullet$]
    \item After inserting or deleting a node
    \begin{enumerate}[--]
        \item Decrease \emph{timer} by 1 for all ancestor nodes.
        \item Then, between ancestor nodes whose \emph{timer} is 0, find the one with least depth.
    \end{enumerate}
    \item During a subtree's rebuild
    \begin{enumerate}[--]
        \item For all of its nodes, reset the node's \emph{timer}.
    \end{enumerate}
\end{enumerate}

In the following, we will call that an insert/delete operation was ``successful'' if it actually inserted/deleted a node.

Now, in the following Section~\ref{sec:insert} to \ref{sec:rebuild}, we see how inserts, deletes, and rebuilds are implemented in our method. We focus on the small changes we add to common implementations of them in naive BSTs, and these changes are also marked on the example code.
\subsection {Insert}\label{sec:insert}
Common implementations of insert operations in naive BSTs use a recursive top-down traversal, briefly explained in the following steps (1) to (3).
\begin{enumerate}[(1)]
 \item Starting from the root, recursively visit nodes in order of increasing depth to find where to insert the node.
 \item Actually insert a new node, or discover that a node with the same key exists.
 \item Go back up the tree, revisiting nodes previously visited but in opposite order.
\end{enumerate}

Insert operations in our method starts from here but adds very small changes. We add a variable \emph{rebuildTarget}, which will store the root of the subtree we want to rebuild. Also, we only change step (3) by adding a short, simple code. Algorithm~\ref{alg:insert} shows the resultant implementation with the change to step (3) marked as (a) and explained in the following.

\begin{algorithm}[ht]
\textsc{Insert}($root, key, result$)
\caption{Insert}
\label{alg:insert}
\begin{algorithmic}
    \IF{$root = null$}
        \STATE $result \leftarrow$ true
        \RETURN new node with key $key$
    \ELSIF{$key < root.key$}
        \STATE $root.left \leftarrow$ \textsc{Insert}($root.left, key, result$)
    \ELSIF{$key > root.key$}
        \STATE $root.right \leftarrow$ \textsc{Insert}($root.right, key, result$)\tikzmark{right1}
    \ELSE
        \STATE $result \leftarrow$ false
        \RETURN $root$
    \ENDIF
    \STATE \tikzmark{left1}\tikzmark{ahead1}\bigbox{left1}{right1}{ahead1}
    \IF{\tikzmark{top1}$result =$ true}
        \STATE $root.timer \leftarrow root.timer - 1$
        \IF{$root.timer = 0$}
            \STATE $rebuildTarget \leftarrow root$\tikzmark{bottom1}
        \ENDIF
    \ENDIF
    \RETURN $root$
\end{algorithmic}
\addNote{top1}{bottom1}{right1}{0.5em}{(a)}
\end{algorithm}

\begin{enumerate}[(a)]
 \item We first check whether the insert operation was successful. If it was, then decrease the \emph{timer} of the current node by one, and if the \emph{timer} is now zero, then set \emph{rebuildTarget} to this node.
\end{enumerate}

Note that the extra code is very explicit in meaning and purpose.

Figure~\ref{fig:insert} shows an example of inserting an integer key ``1'' to a tree, with \emph{timer} abbreviated to \emph{t}. (i) shows the tree after steps (1) and (2), and (ii) shows the tree after our modified step (3), with \emph{rebuildTarget} marked as a black node.

\begin{figure}
        \begin{subfigure}[(i)]{0.5\textwidth}
                \includegraphics[width=\linewidth]{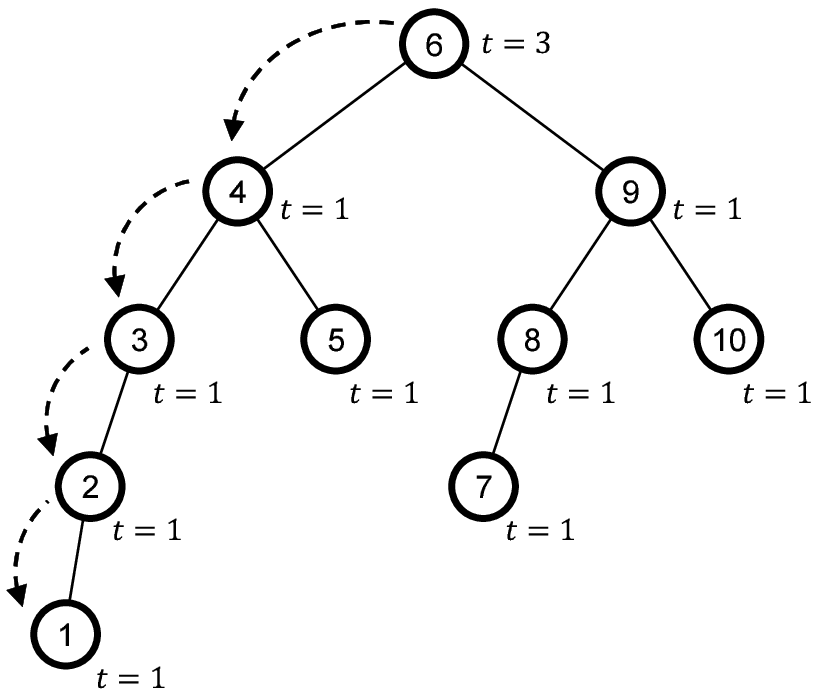}
                \caption{New node with key ``1'' inserted.\newline}
                \label{fig:insert1}
        \end{subfigure}%
        \begin{subfigure}[(i)]{0.5\textwidth}
                \includegraphics[width=\linewidth]{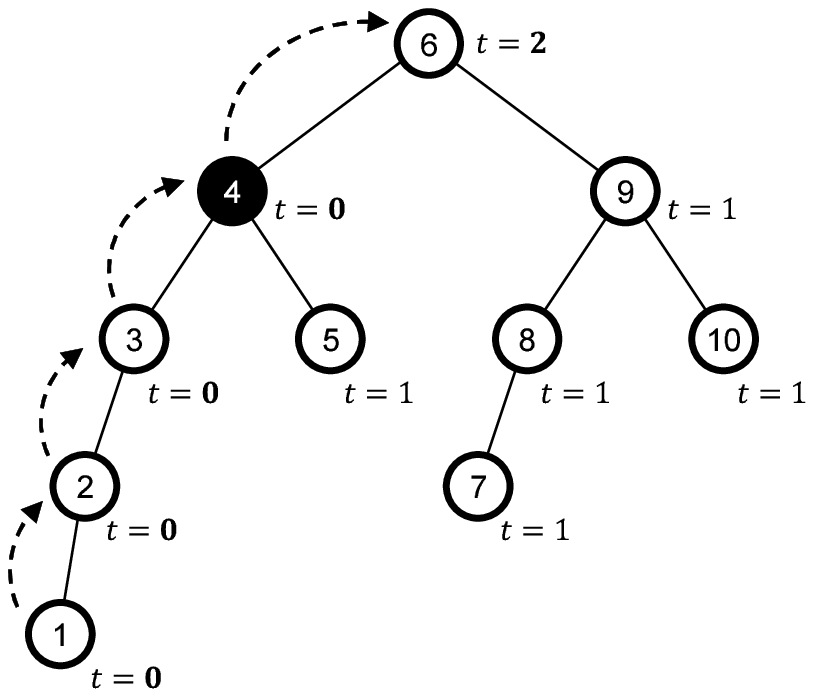}
                \caption{\emph{timer} values decreased and\newline \emph{rebuildTarget} designated.}
                \label{fig:insert2}
        \end{subfigure}%
        \caption{Inserting ``1'' in a tree.}\label{fig:insert}
\end{figure}

Additionally, other methods such as passing \emph{rebuildTarget} as a second return value can be used instead if preferred.

After designating the \emph{rebuildTarget} node, we just need to rebuild its rooted subtree.

\subsection {Delete}
For delete operations, we only need the exact same changes we added in insert operations. As in insert operations, common implementations of delete operations in naive BSTs also use a recursive top-down traversal, briefly explained in the following steps (1) to (3).
\begin{enumerate}[(1)]
 \item Starting from the root, recursively visit nodes in order of increasing depth to find the node we want to delete. If the node is found but has two child nodes, exchange it into a problem of deleting its inorder successor and continue.
 \item Actually delete a node that has zero or one child nodes, or discover that a node with such key does not exist.
 \item Go back up the tree, revisiting nodes previously visited but in opposite order.
\end{enumerate}

Delete operations in our method also starts from here and adds the identical changes we did in insert operations. We add the variable \emph{rebuildTarget} and only change step (3) by adding the exact same code. Algorithm~\ref{alg:delete} shows the resultant implementation with the change to step (3) marked as (a) and explained in the following.

\begin{algorithm}
\textsc{Delete}($root, key, result$)
\caption{Delete}
\label{alg:delete}
\begin{algorithmic}
    \IF{$root = null$}
        \STATE $result \leftarrow$ false
        \RETURN $root$
    \ELSIF{$key < root.key$}
        \STATE $root.left \leftarrow$ \textsc{Delete}($root.left, key, result$)
    \ELSIF{$key > root.key$}
        \STATE $root.right \leftarrow$ \textsc{Delete}($root.right, key, result$)\tikzmark{right}
    \ELSE
        \IF{both $root.left$ and $root.right$ is not $null$}
            \STATE $root.key \leftarrow$ \textsc{GetMin}($root.right$)$.key$
            \STATE $root.right \leftarrow$ \textsc{Delete}($root.right, root.key, result$)\tikzmark{right2}
        \ELSE
            \STATE $result \leftarrow$ true
            \IF{$root.left \neq null$}
                \RETURN $root.left$
            \ENDIF
            \RETURN $root.right$
        \ENDIF
    \ENDIF
    \STATE \tikzmark{left2} \tikzmark{ahead2} \bigbox{left2}{right2}{ahead2}
    \IF{\tikzmark{top2}$result =$ true}
        \STATE $root.timer \leftarrow root.timer - 1$
        \IF{$root.timer = 0$}
            \STATE $rebuildTarget \leftarrow root$\tikzmark{bottom2}
        \ENDIF
    \ENDIF
    \RETURN $root$
\end{algorithmic}
\addNote{top2}{bottom2}{right2}{0.5em}{(a)}
\end{algorithm}

\begin{enumerate}[(a)]
 \item We first check whether the delete operation was successful. If it was, then decrease the \emph{timer} of the current node by one, and if the \emph{timer} is now zero, then set \emph{rebuildTarget} to this node.
\end{enumerate}

After designating the \emph{rebuildTarget} node, we just need to rebuild its rooted subtree, just as in insert operations.

\subsection {Partial rebuild}
\label{sec:rebuild}
Some trees maintain balance by using \emph{partial rebuilds}, and this method was introduced by Overmars and van Leeuwan \cite{Overmars:1987:DDD:535500, Overmars:1982:DMD:2697726.2697937}. A partial rebuild rebalances a subtree into a perfectly balanced one, where we call a subtree is perfectly balanced if the size of a node's left subtree and right subtree differ by one at most for every node in it. Various partial rebuilding algorithms exist. \cite{BentleyMultiBST, Chang84, Martin72, Stout86}. In our case, an insert/delete operation may include a partial rebuild.

Just as tree rotations, partial rebuilds are not necessarily included in naive BSTs. However, we have a simple partial rebuilding algorithm that uses only a basic tree traversal and a simple recursive algorithm. By its simplicity, this method is commonly used, and the following is a brief explanation of it in steps (1) to (2).
\begin{enumerate}[(1)]
 \item Copy all nodes of the subtree to a linear array in sorted order. This is done by an in-order traversal over the subtree. At each visit, simply push the node at the back of the linear array.
 \item Make a perfectly balanced tree using the linear array in a divide-and-conquer sense. This is done by a recursive algorithm that returns the median node of the array after setting that node's left/right child to the return value of the recursive call for the left/right half of the array excluding the median.
\end{enumerate}

Not only is this common method simple to implement, but it also only needs a very small change. During (2), when we visit each node, we just need to reset the node's \emph{timer}. Algorithm~\ref{alg:rebuild1} to \ref{alg:rebuild2} is an example implementation, with (a) the part that resets the \emph{timer}.

\begin{algorithm}
\textsc{CopyToArray}($node, array$)
\caption{Copy all nodes to a linear array.}
\label{alg:rebuild1}
\begin{algorithmic}
    \IF{$node \neq null$}
        \STATE \textsc{CopyToArray}($node.left, array$)
        \STATE $array.push(node)$
        \STATE \textsc{CopyToArray}($node.right, array$)
    \ENDIF
\end{algorithmic}
\end{algorithm}

\begin{algorithm}
\textsc{BuildTree}($array, begin, end$)
\caption{Build a perfectly balanced tree from a linear array.}
\label{alg:rebuild2}
\begin{algorithmic}
    \IF{$begin > end$}
        \RETURN $null$
    \ENDIF
    \STATE $m \leftarrow \left\lfloor (begin+end)/2\right\rfloor$
    \STATE $root \leftarrow array[m]$
    \STATE $root.left \leftarrow$ \textsc{BuildTree}($array, begin, m-1$)\tikzmark{right3}
    \STATE \tikzmark{left3}\tikzmark{ahead3}$root.right \leftarrow$ \textsc{BuildTree}($array, m+1, end$)\smallbox{left3}{right3}{ahead3}
    \STATE \tikzmark{top3}$root.timer \leftarrow \max\left(1, \left\lfloor k\cdot(end-begin+1)\right\rfloor\right)$\tikzmark{bottom3}
    \RETURN $root$
\end{algorithmic}
\addNote{top3}{bottom3}{right3}{0.2em}{(a)}
\end{algorithm}
Many other partial rebuild algorithms also only need a very small change.

%////////////////////////////////////////////////////////////////////////////////
\section {Proof of Logarithmic Height and Amortized Logarithmic Costs}
\label{sec:proof}
In this section, we will show that the tree has logarithmic height, which means that search operations or unsuccessful insert/delete operations have logarithmic costs. Also, we will show that successful insert/delete operations have amortized logarithmic costs.

\subsection{Logarithmic height}
Many trees that use partial rebuilding maintain logarithmic height by using a constant lower bound (or upper bound) on the quotient of a subtree's subtree's weight (or size) to the subtree's own weight (or size). In our tree, such quotient is always larger than $\frac{1-2k}{2-2k}$ when $0<k<0.5$, but when $k\geq 0.5$, such quotient is not bounded by a constant in $(0,1)$. Because we do not use a balance criterion, our tree's shape is flexible, just as in general balanced trees, and even the root of a not-so-small tree may have only one child node.

Hence, instead of comparing intrinsic matters, we compare potential matters. In the following, we will show that the quotient between upper bounds on sizes (instead of just sizes) is upper bounded by a constant. This also offers a (mostly) tighter upper bound on height.

First, note that when we insert a new node, the \emph{timer} gets ``set'' for the subtree rooted by it, and when we rebuild a subtree, the \emph{timer} gets ``reset'' for all subtrees of it. Based on a subtree's last \emph{timer} (re)set, if we refer to an insertion or deletion as an ``update'', we can divide a subtree into two cases.
\begin{lemma}
All subtrees are always in one of the following two cases.
\begin{enumerate}
    \item No updates happened to it since its last timer (re)set.
    \item At least one update happened to it since its last timer (re)set.
\end{enumerate}
\label{lemma:2cases}
\end{lemma}

We will often use this kind of division in the following. Note that since a subtree is perfectly balanced right after its \emph{timer} (re)set, we mostly only need to focus on subtrees in the second case. Also, note the following.

\begin{lemma}
If a subtree of size less than $\frac{2}{k}$ went through no updates since its last timer (re)set, then its timer gets reset after an update in it, assuming it is still non-empty.
\label{lemma:toosmallT}
\end{lemma}

\begin{lemma}
For a perfectly balanced subtree of size $n$, the size of its subtree is no more than $\frac{n}{2}$, if exists.
\label{lemma:perfsize}
\end{lemma}

Using these, we can see the following.
\begin{theorem}
Starting from a subtree whose timer was (re)set, if its size cannot exceed $n$ by doing no more than $x$ any updates, then its subtree's size also cannot exceed $\left(\frac{1+2k}{2+2k}\right)n$ by doing no more than $x$ any updates.
\label{thm:quotient}
\end{theorem}
\begin{proof}
Suppose that it can exceed. Then, using Lemma~\ref{lemma:2cases}, if this happened in a subtree $T$'s subtree, we can express that this happened after doing $u$ updates (that did not include $T$'s \emph{timer} reset) starting from one of $T$'s \emph{timer} (re)sets.

Call $n_0$ the previous size of $T$ right after ``this'' \emph{timer} (re)set happened. Then, $0\leq u < \max(1, \lfloor kn_0 \rfloor)$, and because of Lemma~\ref{lemma:perfsize}, $\frac{n_0}{2}+u > \left(\frac{1+2k}{2+2k}\right)n$ must be true. Note that we must have $n\geq n_0+u$. (Otherwise, $T.size>n$ is possible.)

However, note the following.
\begin{enumerate}[(i)]
    \item If $n_0 < \frac{2}{k}$, then $u=0$, and we have
    \begin{equation*}
        \left(\frac{1+2k}{2+2k}\right)n \geq \left(\frac{1+2k}{2+2k}\right)n_0 \geq \frac{n_0}{2}.
    \end{equation*}
    \item If $n_0 \geq \frac{2}{k}$, then $u\leq \lfloor kn_0 \rfloor-1 \leq kn_0$.\\
    However, we said that
    \begin{equation*}
        \frac{n_0}{2}+u > \left(\frac{1+2k}{2+2k}\right) n \geq \left(\frac{1+2k}{2+2k}\right)(n_0+u),
    \end{equation*}
    and this means $u>kn_0$, which contradicts.
\end{enumerate}

Therefore, we conclude that it is impossible.
\end{proof}

Also, it is not hard to see that $\left(\frac{1+2k}{2+2k}\right)$ is the smallest constant that we can guarantee in Theorem~\ref{thm:quotient}. 

Now, we finally have the following.
\begin{theorem}
The height of a non-empty tree with size $n$ is no more than $\left\lfloor \log_{\frac{2+2k}{1+2k}}{\frac{n}{2-2k}} \right\rfloor+1$. That is, the height of the tree is $O(\log n)$.
\end{theorem}
\begin{proof}
We divide the tree into two cases using Lemma~\ref{lemma:2cases}.

A tree in the first case of Lemma~\ref{lemma:2cases} is perfectly balanced. That is, all subtrees inside it got their \emph{timer} (re)set, and the tree's height is no more than $\lfloor \log_{2}{n}\rfloor$.

Now, we check for non-empty trees in the second case of Lemma~\ref{lemma:2cases}. Call $n_0$ the previous size of the tree right after its last \emph{timer} reset, and note that $n_0\geq \frac{2}{k}$ by Lemma~\ref{lemma:toosmallT}. Also, note that we are considering trees that went through no more than $\lfloor kn_0 \rfloor-1$ any updates since its last \emph{timer} reset.

If the number of insertions the tree went through since its last \emph{timer} reset is $u$, then its left/right subtree's size cannot exceed $\frac{n_0}{2}+u$ ($\because$ Lemma~\ref{lemma:perfsize}), and we have $n \geq n_0-\lfloor kn_0 \rfloor+2u+1$. Then, we can easily see that the quotient between $\frac{n_0}{2}+u$ and $n$ is less than $\frac{1}{2-2k}$. Now, using Theorem~\ref{thm:quotient}, we can easily see that the size of a subtree rooted by a node at depth $d$ is always no more than
\begin{equation*}
    \left(\frac{1+2k}{2+2k}\right)^{d-1}\frac{n}{2-2k}.
\end{equation*}
Therefore, we conclude  that the height is no more than $\left\lfloor \log_{\frac{2+2k}{1+2k}}{\frac{n}{2-2k}} \right\rfloor+1$.
\end{proof}

Note that a different $k$ value leads to a different upper bound on height. In our tree, we can say that a smaller $k$ leads to more frequent rebuilds and a smaller height.

\subsection{Amortized costs of insert and delete operations}
\begin{theorem}
\label{thm:constUpper}
Assume that rebuilding a subtree $T$ costs no more than $T.size$. If a subtree's timer reached 0, then the quotient between its rebuild cost and $u$, the number of updates it went through since its last timer (re)set, is less than $\left(\frac{2}{k}+1\right)$. That is, the rebuild cost is less than $\left(\frac{2}{k}+1\right)u$.
\end{theorem}
\begin{proof}
Call $n$ the size of the subtree we want to rebuild, and $n_0$ the size of it before it went through $u$ such updates. Then, $n \leq n_0+u$, and $u = \max\left(1, \left\lfloor kn_0 \right\rfloor \right)$. Now, we check two cases.
\begin{enumerate}[(i)]
    \item If $n_0<\frac{1}{k}$, then $u=1$ and $n < \frac{1}{k}+1$.
    \item If $n_0\geq\frac{1}{k}$, then $u=\left\lfloor kn_0 \right\rfloor$ and $n\leq n_0+\left\lfloor kn_0 \right\rfloor$.\\
    Then, for the quotient between the rebuild cost and $u$,
    \begin{equation*}
        \frac{n}{u} \leq \frac{n_0+\left\lfloor kn_0 \right\rfloor}{\left\lfloor kn_0 \right\rfloor} =\frac{n_0}{\left\lfloor kn_0 \right\rfloor}+1 < \frac{2}{k}+1.
    \end{equation*}
\end{enumerate}
Therefore, in both cases, the quotient is less than $\left(\frac{2}{k}+1\right)$.
\end{proof}

This means that if we add $\left(\frac{2}{k}+1\right)$ extra credit at all ancestor nodes of the inserted/deleted node after each insertion/deletion, a subtree's root will always have enough credit to pay the subtree's rebuild cost by the time its \emph{timer} reached 0. Since our tree's height is $O(\log n)$, so is the number of ancestor nodes of the inserted/deleted node, and therefore, we conclude with the following.
\begin{theorem}
Our tree can be maintained with $O(\log n)$ amortized cost per successful insert/delete operation, where $n$ is the size of the tree the operation happened.
\end{theorem}

%////////////////////////////////////////////////////////////////////////////////
\section {Discussion of Simplicity}
\label{sec:discussions}
As previously mentioned, a balanced BST is a fundamental data structure, but unfortunately, it is also infamous for its cumbersome implementation. Note that nowadays, so many people from all kinds of fields need and use computer science. However, not only are common methods said to be too complex for average programmers, but they are also often neglected or replaced by poor methods.

Also, in many cases, people may not need cutting-edge performance but rather just want to quickly and simply implement their idea. Moreover, when performance is not a high priority, people may prefer the simplest method in many cases.

This strongly emphasizes why simple methods are important, just as much other previous research already did.

\subsection*{Our method}
Simplicity is one of the key virtues when implementing algorithms, and our method ultimately aimed for maximum simplicity. Our method offers the following.

First, our method does not use any balance criterion, so we can minimize the time we consume to understand a tree. Before implementing a tree, one must first thoroughly understand the tree to create a correct, error-less implementation. Not to mention this is the hardest part in most cases. When using a balance criterion, one would usually need to understand when we call a tree is ``unbalanced'', how to detect these, and specifically what to do to rebalance the tree. Instead, we use a highly intuitive method that can be briefly explained in one sentence, and we just need to remember that we will schedule rebuilds.

Second, our method can be implemented by adding only minimal code and a simple partial rebuilding algorithm to a naive BST. The code we add is also highly simple and explicit.

Third, we use the exact same method and extra code for both insert and delete operations. To cite Sen and Tarjan \cite{SenNoDel}, note that rebalances after deletions are generally more complicated than after insertions, and to make things worse, rebalances after deletions are often not even included in textbooks or databases. Though general balanced trees may be an exception, it still uses a whole different method for deletes.

Finally, we argue that the use of partial rebuilds offers a more abstract and simple method compared to tree rotations. When using rebalances that use tree rotations, we need to consider all cases of the tree, implement various types of complicated rotations, and determine the exact type needed for each case. Implementing all of this is error-prone and hard to do without additional reference or figures, especially for novice. Therefore, it is often criticized to be too complex. \cite{AnderssonAA, MunroDSL} On the other hand, partial rebuilding is a \emph{generalized} algorithm. In all cases, all we need to do is designate the rebuild target. Not to mention that rebuilds itself also can be implemented by only using simple algorithms.

%We believe our method will be highly useful as a simple and short solution for novice or when amortized logarithmic costs are enough. Also, in cases where performance is not a high priority or in search-intensive cases, our method will be suitable.

\section{Conclusion}
More than simplifying a balance criterion, it was shown that we can maintain a balanced BST without using any balance criterion. Ultimately, this offers a simple and explicit method that only needs to add minimal code and a simple partial-rebuilding algorithm to a naive BST. For novice or when amortized logarithmic costs are enough, our method will be suitable as a simple solution.

\bibliographystyle{abbrv}
\bibliography{ref.bib}

\end{document}